\documentclass{article}
\usepackage[papersize={8.5in,11in},margin=1in]{geometry}

\usepackage{booktabs} 

\usepackage{libertine}
\usepackage{newtxmath}
\usepackage{zi4}
\usepackage[colorlinks]{hyperref}
\hypersetup{linkcolor=black,filecolor=black,citecolor=black,urlcolor=black}

\usepackage{amsthm}
\usepackage{amsmath}
\usepackage{amssymb}
\usepackage{mathrsfs}
\usepackage{hyperref}
\usepackage[square,sort&compress]{natbib}
\usepackage{graphicx}
\usepackage{subcaption}
\usepackage{balance}
\usepackage{aliascnt}
\usepackage{subcaption}

\newtheorem{theorem}{Theorem}
\newtheorem*{theorem*}{Theorem}

\newaliascnt{lemma}{theorem}
\newaliascnt{corollary}{theorem}
\newaliascnt{claim}{theorem}
\newaliascnt{observation}{theorem}
\newaliascnt{definition}{theorem}
\newtheorem{lemma}[lemma]{Lemma}
\newtheorem{observation}[observation]{Observation}
\newtheorem{claim}[claim]{Claim}
\newtheorem{corollary}[corollary]{Corollary}
\newtheorem{definition}[definition]{Definition}
\aliascntresetthe{lemma}
\aliascntresetthe{corollary}
\aliascntresetthe{claim}
\aliascntresetthe{observation}
\aliascntresetthe{definition}

\newcommand\R{\mathbb{R}}

\newcommand{\Infer}{\Longrightarrow}
\newcommand{\ud}{\mathrm{d}}
\newcommand{\area}{\mathsf{area}}
\newcommand{\core}{\mathsf{core}}
\newcommand{\tail}{\mathsf{tail}}
\newcommand{\btau}{\boldsymbol{\tau}}
\renewcommand{\t}{\boldsymbol{t}}
\newcommand{\p}{\boldsymbol{p}}
\renewcommand{\H}{\mathcal{H}}
\renewcommand{\P}{\mathcal{P}}

\begin{document}
\title{The Matthew Effect in Computation Contests: \\
High Difficulty May Lead to 51\% Dominance}
\author{
Yulong Zeng\\
Nebulas\\
\texttt{yulong.zeng@nebulas.io}
\and
Song Zuo\\
Google Research\\
\texttt{szuo@google.com}
}
\date{}

\maketitle

\begin{abstract}
  We study the computation contests where players compete for searching a solution
to a given problem with a winner-take-all reward. The search processes are
independent across the players and the search speeds of players are proportional
to their computational powers. One concrete application of this abstract model
is the mining process of proof-of-work type blockchain systems, such as Bitcoin.
Although one's winning probability is believed to be proportional to his
computational power in previous studies on Bitcoin, we show that it is not the
case in the strict sense.

Because of the gaps between the winning probabilities and the proportions of
computational powers, the Matthew effect will emerge in the system, where the
rich get richer and the poor get poorer. In addition, we show that allowing the
players to pool with each other or reducing the number of solutions to the
problem may aggravate the Matthew effect and vice versa.

\end{abstract}

\section{Introduction}\label{sec:intro}

In a computation contest, a group of players compete with each other on solving
a given computational problem and only the player who first finds a solution to
this problem will receive a fixed reward. To solve the problem, the most
efficient way is to execute a random brute force search across the candidate
set, so the computational power is a critical resource for each player to
increase the probability of winning this contest. Practical examples abound,
ranging from tuning the parameters for a deep neural network, selecting random
seeds for reinforcement learning to recovering the spatial structure of proteins
and discovering new drugs for cancers \citep{cooper2010predicting,mnih2015human,horowitz2016determining,silver2017mastering,silver2017zero}.

The mining process of {\em Bitcoin} \citep{nakamoto2008bitcoin} is a concrete
example that perfectly matches the computation contest in terms of mathematics.
The Bitcoin system is maintained as a chain of blocks and it awards the builder
of each block by a certain amount of BTC, which is about $12.5$ BTC ($\approx
79,400$ USD) at the time this paper is written. To build (or mine) a block, one
must find a block header (a fixed-length bit string) whose hash value meets a
certain criteria and announce it before other competitors. Because it is hard to
invert the secure hash function, the most efficient way is to enumerate (in an
arbitrary order) and verify the candidate block headers. This is exactly the
case as we described in the computation contest and the mining process is now an
extremely fierce competition --- the total computational power of the community
is roughly $48$E ($4.8 \times 10^{19}$) hash per second.\footnote{See \url{https://www.blockchain.com/en/charts/hash-rate}.}

To participate such a vast contest, understanding the winning probability or the
expected reward is undoubtedly important to the players. So far, the winning
probability of each player is believed to be proportional to his computational
power (or his speed of verifying one candidate) \citep{fisch2017socially}. Their
argument is briefly as follows: For each player, the wins follow a Poisson
process with a strength proportional to his computational power and hence the
expected reward is also proportional to his computational power.

However, in the strict sense, this conclusion is incorrect due to the nature of
draws without replacement in the computation contests: Conditioned on no
solution has been found yet, the more candidates one has already checked, the
higher probability the next candidate will be a solution. To see this
concretely, take the following extreme case as an example. There are only two
players with computational powers being $2$ and $1$. When there is only one
solution to the given problem, the winning probabilities of them are in fact
$3/4$ and $1/4$, instead of $2/3$ and $1/3$. In \autoref{sec:single}, we will
demonstrate the analysis in details. In general, our analysis shows that the gap
between one's winning probability and the proportion of his computational power
increases as the number of solutions of the problem decreases.

Such gaps, if not negligible, will lead to the Matthew effect in the system
where the players repeatedly participate the computation contest. In practice,
the reward received now often results in increments of computational powers in
the future, for example, research groups can attract more researcher and funds
by their results and Bitcoin miners can purchase more machines using BTCs. The
gaps, as we will show in this paper, offer higher increments for those with
relatively high computational powers and makes {\em the rich (relatively) richer
and the poor (relatively) poorer}. Eventually, the entire system could be
dominated by a single player, which is undesirable in most applications. The
converging speed could be accelerated if the players are allowed to form groups
or {\em pools} where they can jointly search for the solution and split the
rewards.

Our analysis of the computation contest then suggests that keeping the number of
solutions large is critical to suppress the Matthew effect in the system.

\subsection{Our contribution}
In this paper, we make the following contributions, where all the results apply
to general multiple-solution cases:
\begin{itemize}
  \item We give the winning probabilities of the players in closed-form, both
        for the single-solution case (\autoref{thm:winning-prob}) and the
        multiple-solution case (\autoref{thm:multi-winning-prob}).
  \item In general, the probabilities are {\em not} proportional to the
        computational powers of the players. In extreme cases, this may lead to
        {\em the Matthew effect} in terms of computational powers: ``the rich
        get richer and the poor get poorer''.
  \item We introduce a pool choosing game between the players, where each of
        them can either specify one pool to join or keep himself independent. We
        analyse the incentives of the players under this game
        (\autoref{sec:pooling}) and identify a non-trivial Nash equilibrium
        (\autoref{thm:equilibrium}).

        In particular, in the Nash equilibrium, there is always a pool (or a
        single player) who has more than half of the total computational power.
        It implies that the Matthew effect could become even stronger when
        pooling is allowed. For the application of Bitcoin, it also means the
        $51\%$ attack\footnote{The $51\%$ attack is a concept in PoW blockchain
        systems \citep{tschorsch2016bitcoin}. It states that if one player (or a
        pool) has computational power more than half of the total computational
        power in the system, then he/she is able to manipulate the blockchain.}
        could emerge in this case.
  \item Last but not least, we show that increasing the number of solutions will
        mitigate the Matthew effect. This is because as the number of solutions
        grows, the winning probabilities will gradually get close to being
        proportional to computational powers, where the gap between the
        probabilities and the proportions can be roughly bounded by the order of
        $n^{-1}$ (\autoref{thm:probmulti}).
\end{itemize}

\subsection{Related work}
The blockchain and Bitcoin are first presented in 2008 under the pseudonym
\citet{nakamoto2008bitcoin} and widely concerned in the past decade. Thousands
of papers follow the work by \citet{nakamoto2008bitcoin} with research areas
including cryptography, game theory, and distributed systems. The book
``Mastering Bitcoin'' \citep{antonopoulos2014mastering} introduces the basic
knowledge about Bitcoin for beginners and \citet{tschorsch2016bitcoin}
summarize some early studies and progresses of the Bitcoin techniques.

The game-theoretical analysis of the mining systems, where each miner performs
as a self-interested player to maximize his reward, is one of the central topics
of blockchain researches and most relevant to this work.
\citet{rosenfeld2011analysis} introduces the attack called ``pool-hopping'',
where the miner can work in one pool at the beginning of one round and hop to
another pool when the current round is longer than a threshold, and analyzes the
advantages and disadvantages under several reward mechanisms of mining pools.
\citet{eyal2018majority} introduce the attack called ``selfish mining''. In such
an attack, when a miner successfully mines a block, he can choose to not publish
it and continue mining based on that block, called the private chain. Results
show that a rational miner will follow the selfish mining strategy as long as
the selfish miner controls at least $1/3$ of the total computational power.
\citet{kiayias2016blockchain} consider a mining game with similar attacks, where
the model structure is a tree of blocks and a block becomes valid when its
height is larger than a given threshold. \citet{eyal2015miner} proposes an
attack that a mining pool can send some ``spies'' to other pools and share the
utilities of the spies obtained from these pools. \citet{liu2018evolutionary}
consider how a miner should choose a mining pool of which the required
computational power is fixed. \citet{fisch2017socially} show that the geometric
pay pool achieves the optimal social welfare based on the power utility
function.

In the famous economic model, R$\&$D race
\citep{harris1987racing,canbolat2012stochastic}, a similar setting is analyzed,
where each agent can choose his effort rate to the winner-takes-all competition.
However, the expected profits are still assumed to be linear to the agents'
effort rates with additional costs. To the best of our knowledge, no literature
drops the assumption about the proportional winning probability and this paper
is the first to mathematically consider the winning probabilities in first
principles.

\section{Preliminaries}\label{sec:prelim}

Consider the following model of an abstract computation contest where the
problem is solved via random trials:
\begin{itemize}
  \item $m$ players compete with each other in solving a problem;
  \item The problem is to find a solution of $h(v) = 0$ from a candidate set
        $V$, where $|V| = N$ is finite;
  \item There are $n \geq 1$ solutions to the problem, but only the player who
        first finds one solution receives a reward $1$.
\end{itemize}
In particular, there is no trick for solving the problem in the sense that
enumerating and verifying all candidates from $V$ in an arbitrary order is the
best algorithm. Formally, let $h$ be randomly drawn from a ground set $\H$
according to some distribution, such that for any $v \in V$, the event $h(v) =
0$ is independent of the function values of any subset of $V \setminus \{v\}$.
In addition, the probability of $h(v) = 0$ is the same for any $v \in V$ and the
computational resource needed to determine whether $h(v) = 0$ or not is the also
the same for any $v \in V$, both independent of the knowledge about the function
values on $V \setminus \{v\}$.\footnote{Although the assumption here can hardly
be exactly met in real applications, our analysis still applies at some extend.
For example, in the case of blockchain under PoW schemes, the hash function $h$
is fixed and so are its solutions. However, as long as the hash function is
still secure, the algorithmic advantage of solving the problem is negligible and
every one is searching in a random order. As a result, the probability of the
$t$-th candidate $v_t$ being a solution is still a constant.}

Throughout the paper, we refer each $v \in V$ as a {\em candidate} and the ones
such that $h(v) = 0$ as {\em solutions}. In the contest, each player $i$ has a
type $x_i > 0$ called the {\em computational power}, i.e., the number of
candidates he can verify in a unit of time. The player with higher computational
power will be able to test more candidates and of course has higher chance to
win. Here, we assume that each player is enumerating through the candidate set
in a random order independent of others.

\paragraph{Blockchain}
One concrete real-world example of this computation contest is the application
of blockchains under Proof of Work (PoW) schemes. As its name would suggest, a
{\em blockchain} is a chain of many blocks that are accepted by the peers, where
each accepted block must have a valid {\em block header} pointing to the
previous block.

Taking Bitcoin as the example: Each block header has a fixed length of $80$
bytes ($640$ bits) which in general could be divided into two parts: (i) the
fixed part (e.g., the hash value of the previous block header) and (ii) the
adjustable part (e.g., the {\em nonce} and the {\em merkle root}). In other
words, there is a fixed and finite set $V$ containing all block header
candidates.

A block header is valid if its hash value (a $256$-bit integer) is less than a
target threshold, i.e., $2^{208}(2^{16} - 1) / D \approx 2^{224} / D$, where $D$
is a positive number called {\em difficulty}. In this case, the probability of a
randomly selected block header being valid is $2^{224} / D \cdot 2^{-256} =
2^{-32}D$. In the language of our computation contest model, the function value
$h(v)$ of a candidate block header $v$ is $0$, if and only if $v$ is a valid
block header.

Since the hash function is deterministic, for any specific block, the number of
valid block headers (solutions), $n$, is fixed (but unknown). Note that the
fixed part of block header is randomly determined by the previous block and the
system, therefore it can be seen as the random index of the function $h$, given
that the hash function is ideal.

The main computational challenge of building the next block is to find one of
the $n$ valid block headers from the $N$ possible candidates, or equivalently,
solving $h(v) = 0$. In particular, if the hash function resists preimage
attacks, the most efficient algorithm one can hope for is the brute force
search (in arbitrary orders). In the context of blockchain, the procedure
of searching a valid block header is referred as {\em mining} and the speed of
computing the hash function is often referred as the {\em hash rate} or the
computational power.

In the classic PoW scheme, such as Bitcoin, every miner has the equal right to
mine and the one who first announces the next block with a valid block header
will win the reward on that block. The reward is a constant with respect to the
mining procedure. Therefore, the mining process of each block could be
well-modeled by our computation contest, where each miner corresponds to a
player, block headers correspond to candidates, and valid block headers
correspond to solutions.

\paragraph{Pooling}
In practice, the probability of winning the contest could be extremely low for a
single player and hence the variance of the random reward is very high. To avoid
the high variance, it is a common practice that individuals form a {\em pool}
and jointly search for the solution. Once the pool wins, the reward will be
distributed to the players in the pool according to a pre-specified policy.
The most common policy, in the case of Bitcoin, is essentially a proportional
allocation rule, i.e., the reward a player receives is proportional to the
computational power he contributes to the pool.\footnote{As recorded on
\url{https://investoon.com/mining_pools/btc}, the most common reward schemes in
the pools are PROP, PPS, PPLNS, all with expected reward approximately
proportional to their hashrates.} In the same spirit, work by
\citet{bloch1996sequential} analyzes sequential coalition game, splitting the
whole party into subsets, and studies subgame perfect equilibrium.

\medskip
In this paper, we focus on two types of questions under this computation contest
model:
\begin{itemize}
  \item Given the environment and the computational powers, what are the winning
        probabilities of the players? Are the probabilities proportional to
        their computational powers?
  \item If not proportional to their computational powers, what is the optimal
        pooling strategy for the players?
\end{itemize}

\section{The Single-Solution Case}\label{sec:single}

In this section, we start with the basic case with only one solution, i.e., $n =
1$. Later in \autoref{sec:multi}, we will generalize all results for
single-solution cases to multiple-solution cases.

\subsection{Warmup: two players and a single solution}\label{ssec:2single}

As defined in the previous section, the player who find its first solution
earliest wins. To compute the winning probabilities of the players, we use real
numbers in $[0, 1]$ to index the candidates in the order that they are verified
by player $i$.\footnote{Here we simplify the analysis by using real numbers to
index the candidates. One can of course use the fractions $k / N$ with integers
$k < N$ to index the candidates for an even more rigorous analysis. In this
case, our analysis still applies by replacing the integrations with summations.
The only difference here is the break-ties at $t_i / x_i = t_j / x_j$. Their
probabilities are zero under real number indexing but at order $1/N$ for
fraction number indexing, which is also negligible for large $N$.} Since each
candidate in general will have different indices in the enumerating sequences of
different players, putting the indices together, a {\em coordinate} of each
candidate is defined. In particular, let $t_i$ be the index of the first
solution in player $i$'s sequence and then $t_i / x_i$ is the time that player
$i$ first finds a solution. In this case, the player with minimum $t_i / x_i$
wins. Note that the $t_i$'s are random variables and also the source of the
randomness of the outcome.

Let $p_i$ be the winning probability of player $i$, we have
\begin{align*}
  p_i = \Pr[i\text{ wins}]
    = \Pr\left[t_i / x_i \leq \min_{j \neq i} t_j / x_j\right].
\end{align*}

\begin{figure}
  \centering
  \includegraphics[width=0.28\textwidth]{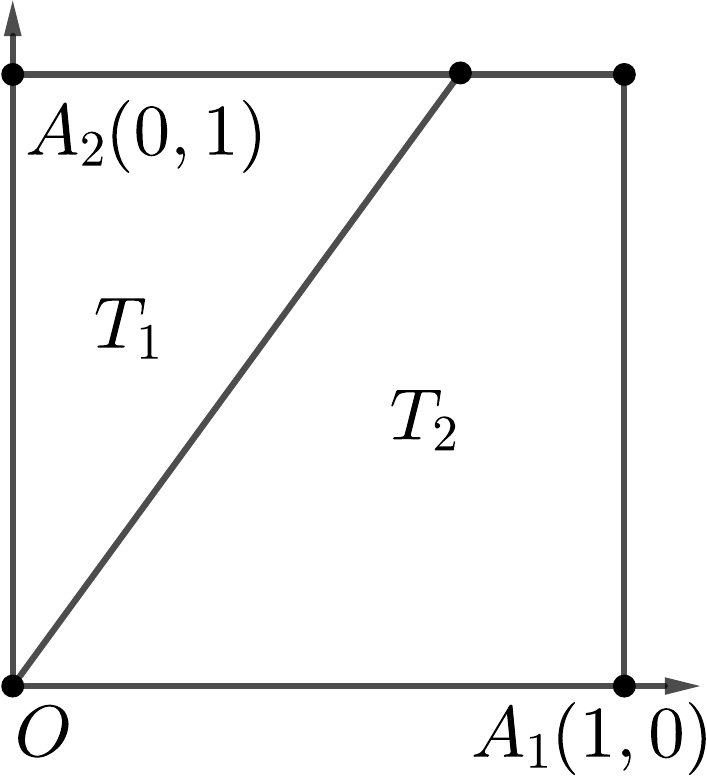}
  \caption{In the two-player case, player $1$ wins if and only if the random
  coordinate $\t$ falls into $T_1$.}\label{fig:2miner}
\end{figure}

For the two-player single-solution case, since there is only one solution, $t_1$
and $t_2$ are independently and uniformly distributed over $[0, 1]$. Denote $\t
= (t_1, t_2)$ and let $T_1$ be the set of $\t$ where player $1$ wins:
\begin{align*}
  T_1 = \{\t \in [0, 1]^2 : t_1 / x_1 \leq t_2 / x_2\}.
\end{align*}
Similarly, we can define $T_2$:
\begin{align*}
  T_2 = \{\t \in [0, 1]^2 : t_2 / x_2 \leq t_1 / x_1\}.
\end{align*}

Because $\t$ is uniformly distributed over $[0, 1]^2$, $p_1 = \int_{T_1} \ud \t
= \area(T_1)$ equals to the area of set $T_1$ (see \autoref{fig:2miner}).
Suppose that $x_1 \leq x_2$, then with some primary calculation, we know that
the slope of the line separating $T_1$ and $T_2$ is $x_2 / x_1$ and hence
\begin{align*}
  p_1 = \area(T_1) = x_1 / 2x_2,~
  p_2 = \area(T_2) = 1 - x_1 / 2x_2,
\end{align*}
which are not proportional to the computational power $x_1$ and $x_2$ (unless
$x_1 = x_2$).

\subsection{$m$ players and a single solution}\label{ssec:nsingle}

  For the general $m$-player case, the winning probability of each player is
  summarized by the following theorem, where $\pi(z)$ and $\pi_{-i}(z)$ are
  defined to simplify the notations:
  \begin{align*}
    \pi(z) = \prod_{i = 1}^m \left(1 - \frac{x_i}{x_m} \cdot z\right), \quad
    \pi_{-i}(z) = \frac{\pi(z)}{1 - \frac{x_i}{x_m} \cdot z}.
  \end{align*}

\begin{theorem}\label{thm:winning-prob}
  Suppose that $x_1 \leq x_2 \leq \cdots \leq x_m$, then the winning probability
  of player $i$ is:
  \begin{align*}
    p_i = \frac{x_i}{x_m}\int_0^1 \pi_{-i}(z) \ud z.
  \end{align*}
\end{theorem}
\begin{proof}
  Recall that
  \begin{align*}
    p_i = \Pr\left[t_i / x_i \leq \min_{j \neq i} t_j / x_j\right]
        = \Pr\bigg[\bigwedge_{j\neq i} t_i / x_i \leq t_j / x_j\bigg].
  \end{align*}
  Let $f_i(\cdot)$ be the probability density function of the random variable
  $t_i$, then
  \begin{align*}
    p_i = \int_0^1 \Pr\bigg[\bigwedge_{j\neq i} t_j \geq \frac{t_i x_j}{x_i}
            \bigg| t_i = t\bigg] f_i(t) \ud t.
  \end{align*}

  Since the orders of the players enumerating through the candidates are
  independent, we have that
  \begin{align*}
    \Pr\bigg[\bigwedge_{j\neq i} t_j \geq \frac{t_i x_j}{x_i}
            \bigg| t_i = t\bigg]
      = \prod_{j \neq i} \Pr\left[t_j \geq \frac{t x_j}{x_i}\right].
  \end{align*}
  In the meanwhile, note that $t_1, \ldots, t_m$ are uniformly distributed,
  hence $f_i(t) = 1$ and
  \begin{align*}
    \Pr\left[t_j \geq \frac{t x_j}{x_i}\right]
    = \max\left\{1 - \frac{t x_j}{x_i}, 0\right\}.
  \end{align*}

  Thus the production is zero for $t \geq x_i / x_m$, which implies
  \begin{align*}
    p_i = \int_0^{x_i / x_m}
            \prod_{j \neq i} \left(1 - \frac{t x_j}{x_i}\right) \ud t.
  \end{align*}

  By substituting $t = x_i z/ x_m$, we complete the proof
  \begin{align*}
    p_i = \frac{x_i}{x_m}\int_0^1
            \prod_{j \neq i} \left(1 - \frac{x_j}{x_m} \cdot z\right) \ud z
        = \frac{x_i}{x_m}\int_0^1 \pi_{-i}(z) \ud z.
  \end{align*}
\end{proof}

According to \autoref{thm:winning-prob}, the winning probabilities $p_i$ are
proportional to the computational powers $x_i$ if and only if the integration
$\int_0^1 \pi_{-i}(z) \ud z$ is a constant for all $1 \leq i \leq m$, which in
general is not true.

In particular, as we will show in \autoref{sec:multi}, the observation above
holds even when there are multiple solutions, i.e., $n > 1$. Nevertheless, as
$n$ goes to infinity, the winning probabilities will gradually get close to
being proportional with respect to the computational powers.

However, the larger the player's computational power is, the higher efficiency
(winning probability divided by computational power) it achieves.

\subsubsection{The Matthew Effect}

The next theorem in fact indicates the Matthew effect in such systems: the
player with higher computational power gets higher return (or higher probability
to win) for each unit of their computational power, i.e., $p_i / x_i \leq p_j /
x_j$.
\begin{theorem}[The Matthew Effect]\label{thm:ratio}
  If $x_i \leq x_j$, then $p_i/x_i \leq p_j/x_j$, where the equality is reached
  if and only if $x_i = x_j$.
\end{theorem}

If both of them invest all the rewards (such as bitcoins) they get to
increase their future computational power, the ratio of $x_j / x_i$ will become
larger and larger over time:

$\Delta x_i$ and $\Delta x_j$ are proportional to $p_i$ and $p_j$, respectively,
hence
\[\Delta x_i / x_i \leq \Delta x_j / x_j, \quad \Infer \quad (x_i + \Delta x_i)
/ (x_j + \Delta x_j) \leq x_i / x_j.\]

In other words, the player with the highest computational power in the system
will gradually take a more and more fraction of the total computational power
and eventually dominate the entire system.

\begin{proof}[Proof of \autoref{thm:ratio}]
  To simplify the notations, write
  \begin{align*}
    \pi_{-ij}(z) = \frac{\pi(z)}{\big(1 - \frac{x_i}{x_m}\cdot z\big)
                                 \big(1 - \frac{x_j}{x_m}\cdot z\big)}.
  \end{align*}

  Therefore, by $x_i \leq x_j \leq x_m$ and $\pi_{ij}(z) \geq 0$, we have that
  $\forall z \in [0, 1]$,
  \begin{align*}
    \left(1 - \frac{x_j}{x_m} \cdot z\right) \cdot \pi_{-ij}(z)
      \leq \left(1 - \frac{x_i}{x_m} \cdot z\right) \cdot \pi_{-ij}(z).
  \end{align*}

  Finally, using \autoref{thm:winning-prob}, we conclude that
  \begin{align*}
    \frac{p_i}{x_i} = \frac1{x_m} \int_0^1
        \left(1 - \frac{x_j}{x_m} \cdot z\right) \cdot \pi_{-ij}(z) \ud z
      \leq \frac1{x_m} \int_0^1
        \left(1 - \frac{x_i}{x_m} \cdot z\right) \cdot \pi_{-ij}(z) \ud z
      = \frac{p_j}{x_j}.
  \end{align*}
  In particular, the equality is reached if and only if $x_i = x_j$.
\end{proof}

\section{Pooling Incentives}\label{sec:pooling}

In this section, we investigate the incentives of players to pool together
assuming that the reward is split proportional to their computational powers. In
particular, we identify a non-trivial Nash equilibrium of the following {\em
pool choosing game}.

\begin{definition}[Pool choosing game]\label{def:game}
  Let $\P_1, \ldots, \P_m$ be $m$ pools. The actions of each player is either to
  choose one from these $m$ pools (say $\P_i$) or to be independent (denoted as
  $\bot$). The utility of each player is the expected winning reward and if a
  pool wins, the reward is split proportional to its members' computational
  powers.
\end{definition}

In the practice of Bitcoin, there are many such pools. The pools have their own
websites, cooperatively mining tools, and reward splitting policies. Usually,
the pools always welcome miners to join so that they can attract more
computational resources. In contrast, each individual miner can of course work
by himself.

As a quick example, the action profile $(\P_1, \ldots, \P_1)$ results in a giant
pool including everyone ($\{1, \ldots, m\}$) and the action profile $(\bot,
\ldots, \bot)$ results in everyone being independent. Although trivial, the
latter is a Nash equilibrium of the game. Because, for example, if player $i$
deviates, then the result is that only player $i$ in a pool and all others being
independent. No difference.

The following theorem then identifies a non-trivial Nash equilibrium for two
cases.

\begin{theorem}[Nash Equilibrium]\label{thm:equilibrium}
  Suppose that $x_m \geq x_{m-1} \geq \cdots \geq x_1$. If $x_m \leq x_1 +
  \cdots + x_{m-1}$, then the action profile $(\P_1, \ldots, \P_1)$ is a Nash
  equilibrium, yielding a giant pool.

  Otherwise, we have $x_m \geq x_1 + \cdots + x_{m-1}$, and the action profile
  $(\P_1, \ldots, \P_1, \bot)$ is a Nash equilibrium, with player $m$ being
  independent and the others in one pool.
\end{theorem}

To prove this theorem, we first need to understand the conditions that a player
has incentive to join a pool (or another player). With the theorems and lemmas
in place, we prove \autoref{thm:equilibrium} at the end of this section.

We also emphasize that the results in this section will later be generalized to
the multiple-solution case in \autoref{sec:multi}.

\subsection{Incentive Analysis}

First of all, an important and immediate observation for pooling is that pooling
is always beneficial when taking the pool as a whole. Formally,
\begin{observation}\label{obs:pool>solo}
  A pool's winning probability is no less than the sum of the winning
  probabilities of its members when they participate the computation contest
  separately.
\end{observation}

This is easy to get by considering a ``pseudo-pool'' of the members, where they
enumerate through the candidates in the same order as they do when participate
the contest separately. Clearly, the winning probability of the ``pseudo-pool''
equals to the sum of the winning probabilities of the members working
separately. Since the players in the ``pseudo-pool'' may repeatedly verify some
of the candidates, the winning probability could be even higher if they simply
skip those already verified by other members in the pool. In particular, this
``higher'' winning probability equals to the winning probability of the real
pool by symmetry.

\begin{theorem}\label{thm:pool}
  For any player $i$, pooling with another player with higher computational
  power than him is always beneficial.
\end{theorem}
\begin{proof}
  Consider a player $i$ pooling with another player $j$ with $x_j \geq x_i$. We
  use $p_{ij}$ to denote the winning probability of the pool. Then by
  \autoref{obs:pool>solo}, we have $p_{ij} \geq p_i + p_j$ and hence the reward
  of player $i$ from the pool is
  \begin{align*}
    p'_i = \frac{x_i}{x_i+x_j} \cdot p_{ij}
         = \frac{p_{ij}}{1+x_j/x_i}
         \geq \frac{p_i + p_i}{1+x_j/x_i}.
  \end{align*}

  Note that by \autoref{thm:ratio}, we have
  \begin{align*}
    p_i / x_i \leq p_j / x_j \quad \iff \quad x_j / x_i \leq p_j / p_i.
  \end{align*}
  Therefore
  \begin{align*}
    p'_i \geq \frac{p_i + p_i}{1+x_j/x_i}
        \geq \frac{p_i+p_j}{1 + p_j / p_i} = p_i.
  \end{align*}

  In other words, player $i$ gets more reward by pooling with player $j$.
\end{proof}

A direct consequence of this theorem is that a pool is stable against single
deviation as long as the largest player's computational power is no more than
half of the computational power of the pool.
\begin{corollary}\label{cor:less_then_50}
  Any individual member of a pool has no incentive to leave if its computational
  power is no more than half of the pool's computational power.
\end{corollary}
\begin{proof}
  Otherwise, once the player left the pool, he would have incentive to rejoin
  the pool by \autoref{thm:pool}, a contradiction.
\end{proof}

In fact, the corollary asserts that any player (or pool) always has incentive to
pool with a player (or pool) that has higher computational power than himself.
In contrast, there would be a tradeoff for the player (or pool) being proposed:
\begin{itemize}
  \item {\bf Pros}: By pooling together, they form a larger pool, which by
        \autoref{thm:ratio}, has a higher efficiency in terms of computational
        power (i.e., higher ratio of reward divided by computational power).
  \item {\bf Cons}: However, the player (or pool) with low computational power
        may actually takes more than he contributes. In this case, the total
        tradeoff for the player with high computational power could be negative
        and he would have incentive to be independent.
\end{itemize}

The following lemma then provides a sufficient condition where the tradeoff is
positive.

\begin{lemma}\label{lem:merge}
  If $x_i + x_j \leq x_m$, then both $i$ and $j$ will be better off in terms of
  the expected reward by pooling.
\end{lemma}
\begin{proof}
  Throughout the proof, we use the notations $\pi_{-ij}(z)$ (defined in the
  proof of \autoref{thm:ratio}) and $p_{ij}$ (the winning probability of the
  pool $\{i, j\}$).

  Without loss of generality, suppose that $x_i \leq x_j \leq x_m$. Then by
  \autoref{thm:pool}, player $i$ gets better off by pooling. In the rest of the
  proof, we show that player $j$ gets better off as well.

  On one hand, according to the proportional allocation rule, the reward of
  player $j$ by pooling is $p'_j = \frac{x_j}{x_i + x_j} \cdot p_{ij}$. Thinking
  the pool $\{i, j\}$ as a player with computational power $x_i + x_j$, then by
  \autoref{thm:winning-prob},
  \begin{align*}
    p_{ij} = \frac{x_i + x_j}{x_m} \int_0^1 \pi_{-ij}(z) \ud z.
  \end{align*}
  Hence
  \begin{align*}
    p'_j = \frac{x_j}{x_i + x_j} \cdot p_{ij}
      = \frac{x_j}{x_m} \int_0^1 \pi_{-ij}(z) \ud z.
  \end{align*}

  On the other hand, the reward of player $j$ without pooling is $p_i$, which by
  \autoref{thm:winning-prob} is,
  \begin{align*}
    p_j = \frac{x_j}{x_m} \int_0^1 \pi_{-j}(z) \ud z.
  \end{align*}
  Note that for any $z \in [0, 1)$, since $x_i < x_m$ and $\pi_{-ij}(z) > 0$,
  \begin{align*}
    \pi_{-ij}(z) > \left(1 - \frac{x_i}{x_m} \cdot z\right) \pi_{-ij}(z)
      = \pi_{-j}(z).
  \end{align*}
  Therefore,
  \begin{align*}
    p'_j > p_j.
  \end{align*}
  In other words, player $j$ gets strictly higher expected reward by pooling with
  player $i$.
\end{proof}

\subsection{Equilibrium Analysis}

Now we are ready to prove \autoref{thm:equilibrium}.

\begin{theorem*}[\autoref{thm:equilibrium} restate]
  Suppose that $x_m \geq x_{m-1} \geq \cdots \geq x_1$. If $x_m \leq x_1 +
  \cdots + x_{m-1}$, then the action profile $(\P_1, \ldots, \P_1)$ is a Nash
  equilibrium, yielding a giant pool.

  Otherwise, we have $x_m \geq x_1 + \cdots + x_{m-1}$, and the action profile
  $(\P_1, \ldots, \P_1, \bot)$ is a Nash equilibrium, with player $m$ being
  independent and the others in one pool.
\end{theorem*}

\begin{proof}[Proof of \autoref{thm:equilibrium}]
  The first case is a direct consequence of \autoref{cor:less_then_50}. Because
  none of them has computational power more than half of the pool, no player has
  incentive to deviate (leave the pool).

  \medskip
  For the second case, we first show that player $m$ has no incentive to
  deviate (or equivalently, join pool $\P_1$).

  We use $x_{-m}$ and $p_{-m}$ to denote the computational power and the winning
  probability of the pool $\{1, \ldots, m-1\}$, respectively. Since $x_m \geq
  x_{-m}$, by \autoref{thm:ratio}, $p_m / x_m \geq p_{-m} / x_{-m}$. Note that
  $p_m + p_{-m} = 1$, hence $p_{-m} = 1 - p_m$. Then we have
  \begin{align*}
    \frac{p_m}{x_m} \geq \frac{1 - p_m}{x_{-m}} \iff
      p_m \geq \frac{x_m}{x_{-m} + x_m},
  \end{align*}
  where the right-hand-side is the reward of player $m$ if he joins the pool. In
  other words, if player $m$ deviates by joining pool $\P_1$, his winning
  probability will decreases.

  Then we prove that any player $i \in \{1, \ldots, m-1\}$ has no incentive to
  deviate (or equivalently, leave pool $\P_1$).

  To show this, we prove that player $i$ has incentive to join pool $\P'_1$,
  when $\P'_1 = \{1, \ldots, i-1, i+1, \ldots, m-1\}$ and player $m$ chooses
  $\bot$. Thinking the pool $\P'_1$ as a player with computational power $x' =
  x_1 + \cdots + x_{i-1} + x_{i+1} + \cdots + x_{m-1}$. Since $x_i + x' \leq
  x_m$, then by \autoref{lem:merge}, player $i$ has incentive to join pool
  $\P'_1$. In other words, player $i$ gets (weakly) higher reward in pool $\P_1$
  than being independent.

  In summary, for the second case, no player has incentive to deviate, hence a
  Nash equilibrium.
\end{proof}

\section{The Multiple-Solution Case}\label{sec:multi}

In this section, we first extend our results of single-solution cases to
multiple-solution cases. Then we show that as the number of solutions goes to
infinity, the winning probabilities become asymptotically proportional to the
player computational powers. In particular, it means that the Matthew effect
gets weaker as the number of solutions increases.

\subsection{Generalization}

We start with generalizing \autoref{thm:winning-prob}.
\begin{theorem}\label{thm:multi-winning-prob}
  Suppose that $x_1 \leq \cdots \leq x_m$, then the winning probability for
  player $i$ in the multiple-solution case is
  \begin{align*}
    p_i = \frac{x_i}{x_m} \int_0^1 \pi_{-i}(z) \cdot \pi(z)^{n-1} \ud z.
  \end{align*}
\end{theorem}

We omit the proof of this theorem but highlight the only two differences with
the proof of \autoref{thm:winning-prob}. Both are because of the definition of
$t_i$, the index of {\em the first solution} in player $i$'s sequence.
Therefore, the probability density function of $t_i$ and the probability that
$t_j \geq tx_j / x_i$ are different with those in the single-solution case:
\begin{align*}
  f_i(t) = n(1-t)^{n-1},
  \Pr\left[t_j \geq \frac{t x_j}{x_i}\right]
    = \max\left\{1 - \frac{t x_j}{x_i}, 0\right\}^n.
\end{align*}

Then by substituting these two terms in the proof of \autoref{thm:winning-prob},
we get \autoref{thm:multi-winning-prob}.

Based on \autoref{thm:multi-winning-prob}, we claim that all previous results
(except \autoref{thm:winning-prob}) also apply to multiple-solution cases:

\begin{claim}[Generalization]\label{clm:gen}
  \autoref{thm:ratio} and \autoref{lem:merge} can be easily reproved for
  multiple-solution cases with \autoref{thm:multi-winning-prob}. Based on this,
  \autoref{thm:equilibrium}, \autoref{obs:pool>solo}, \autoref{thm:pool}, and
  \autoref{cor:less_then_50} are automatically generalized because their proofs
  do not rely on the number of solutions.
\end{claim}

\subsection{Asymptotically Proportional Winning Probabilities}

Consider the probability that the first solution comes after $t_1, \ldots, t_m$
for each of the players. Since the distribution of each $t_i$ is independent, we
have:
\begin{align*}
  \Pr[\btau \geq \t] = \prod_{i=1}^m \Pr[\tau_i \geq t_i]
  = \prod_{i=1}^m (1-t_i)^n,
\end{align*}
hence its density function:
\begin{align*}
  f(\t) = \prod_{i=1}^n f_i(t_i)
    = \prod_{i=1}^n \frac{\partial \left(1 - t_i\right)^n}{\partial t_i}
    = n^m \prod_{i=1}^m (1-t_i)^{n-1}.
\end{align*}

Let $T_i$ be the set of $\t$ that player $i$ wins and $T = T_1 \cup \cdots \cup
T_m = [0, 1]^m$. Then the winning probability of player $i$ is
\begin{align*}
  p_i = \int_{\t \in T_i} f(\t) \ud \t.
\end{align*}

\begin{theorem}[Asymptotically Proportional]\label{thm:probmulti}
  For any fixed number of players, $m$, as the number of solutions, $n$,
  approaches to infinity,
  \begin{align*}
    p_i \longrightarrow \frac{x_i}{x_1 + \cdots + x_m}.
  \end{align*}

  In particular, the difference is vanishing in the order of
  $\widetilde O(n^{-1})$,\footnote{The $\widetilde O$ notation here hides a
  poly-log factor, meaning that the residue term is bounded by $O\left(n^{-1}
  (\ln n)^c\right)$, for some positive constant $c > 0$ when $n$ is
  sufficiently large.} i.e.,
  \begin{align*}
    \left|p_i - \frac{x_i}{x_1 + \cdots + x_m}\right| = \widetilde O(n^{-1}).
  \end{align*}
\end{theorem}
The theorem could be directly derived based on the following
\autoref{lem:probmulti} and \autoref{lem:splxv}.

We use $\rho_i(z)$ to denote the $(m-1)$-dimensional volume of set $\{\t \in T_i
: t_1 + \cdots + t_m = z\}$ and $\rho(z) = \rho_1(z) + \cdots + \rho_m(z)$.

Then \autoref{lem:probmulti} essentially means that $\p_i$ is asymptotically
proportional to $\rho_i(1)$.
\begin{lemma}\label{lem:probmulti}
  For any constant $m \ll n$,
  \begin{align*}
    p_i = \frac{\rho_i(1)}{\rho(1)}
            \cdot n^m\int_0^{n^{-1+\epsilon}} (1-z)^{n-1} \rho(z) \ud z
          + \widetilde O(n^{-1}).
  \end{align*}
\end{lemma}

Furthermore, \autoref{lem:splxv} indicates that $\rho_i(1)$ is proportional to
the computational power of player $i$.
\begin{lemma}\label{lem:splxv}
  \begin{align*}
    \frac{\rho_i(1)}{\rho(1)} = \frac{x_i}{x_1 + \cdots + x_m}.
  \end{align*}
\end{lemma}

\begin{proof}[Proof of \autoref{thm:probmulti}]
  By \autoref{lem:probmulti} and \autoref{lem:splxv}.
\end{proof}

We dedicate the rest of this section to prove \autoref{lem:probmulti} and
\autoref{lem:splxv}.

\paragraph{Proof idea of \autoref{lem:probmulti}}
The proof is done by a careful estimation of $p_i$ using a technique called the
{\em core-tail decomposition} \citep{li2013revenue}. To estimate the probability,
consider the core-tail decomposition of the set $T_i$, i.e., let
\[\core = \{\t : \bar t \leq \alpha\} \quad \text{and} \quad
\core(T_i) = \core \cap T_i,\]
where $\bar t = t_1 + \cdots + t_m$ and $\alpha > 0$ is a parameter to be
determined later. Similarly, define $\tail = [0, 1]^m \setminus \core$ and
$\tail(T_i) = \tail \cap T_i$.
Thus,
\begin{align}\label{eq:core-tail}
  p_i = \int_{\core(T_i)} f(\t) \ud \t + \int_{\tail(T_i)} f(\t) \ud \t.
\end{align}

The following two lemmas would be helpful on simplifying the analysis.
\begin{lemma}\label{lem:core}
  Let $g(\t) = n^m (1 - \bar t)^{n-1}$. If $\alpha < n^{-2/3}$, then
  \begin{align*}
    \int_{\core(T_i)} f(\t) \ud \t = \int_{\core(T_i)} g(\t) \ud \t
      + O(n^{m+1}\alpha^{m+2}).
  \end{align*}
\end{lemma}
\begin{lemma}\label{lem:tail}
  \begin{align*}
    \int_{\tail(T_i)} f(\t) \ud \t \leq e^{-\alpha(n-1)+m\ln n}.
  \end{align*}
\end{lemma}
We postpone the proofs of these lemmas to \autoref{sec:proof}, which are mainly
elaborate treatments of the Taylor series.

\begin{proof}[Proof of \autoref{lem:probmulti}]
  Plugging \autoref{lem:core} and \autoref{lem:tail} into \eqref{eq:core-tail},
  and letting $\alpha = \frac{(m+1)\ln n}{n - 1}$, we have that,
  \begin{align*}
    p_i - \int_{\core(T_i)} g(\t) \ud \t
      = O\left(n^{-1}(\ln n)^{m+2}\right)
      = \widetilde O(n^{-1}).
  \end{align*}
  Hence we remain to show the following equation:
  \begin{align*}
    \int_{\core(T_i)} g(\t) \ud \t
      = \frac{\rho_i(1)}{\rho(1)} \cdot n^m\int_0^{\alpha} (1-z)^{n-1} \rho(z) \ud z,
  \end{align*}
  or equivalently,
  \begin{gather}\label{eq:multi}\tag{$*$}
    \!
    \int_{\core(T_i)} (1 - \bar t)^{n-1} \ud \t
      = \frac{\rho_i(1)}{\rho(1)} \int_0^{\alpha} (1-z)^{n-1} \rho(z) \ud z.
  \end{gather}

  First of all, note that the left-hand-side of \eqref{eq:multi} could be
  rewritten as a double integration over $\core_z(T_i) = \{\t \in \core(T_i):
  \bar t = z\}$ and then for $z$ from $0$ to $\alpha$:
  \begin{align*}
    \int_{\core(T_i)} (1 - \bar t)^{n-1} \ud \t
      = \int_0^{\alpha} \ud z \int_L (1 - z)^{n-1} \ud \core_z(T_i)
      = \int_0^{\alpha} (1 - z)^{n-1} \ud z \int_L \ud \core_z(T_i),
  \end{align*}
  where $\int_L \ud S$ is the Lebesgue integration over a set $S$. In
  particular, $\int_L \ud \core_z(T_i)$ is the $(m-1)$-dimensional volume of the
  set $\core_z(T_i)$, denoted as $\rho_i(z)$.

  Then consider the ratio between $\rho_i(z)$ and the volume of set
  $\core_z(T)$, $\rho(z)$. Since both $\core(T_i)$ and $\core(T)$ are the
  intersections of the $\core$ and positive cones in $\R^m$, hence $\rho_i(z)$
  and $\rho(z)$ must be proportional to $z^{m-1}$ accordingly:
  \begin{align*}
    \rho_i(z) = z^{m-1}\rho_i(1), \quad \rho(z) = z^{m-1} \rho(1).
  \end{align*}

  Finally, we conclude that
  \begin{align*}
    \int_L \ud \core_z(T_i) = \rho_i(z) = \frac{\rho_i(1)}{\rho(1)} \cdot \rho(z),
  \end{align*}
  which then directly implies \eqref{eq:multi}.
\end{proof}

\paragraph{Proof idea of \autoref{lem:splxv}}
To smoothly proceed the analysis, we will start with the cases of $m = 2$ and
$m = 3$.

\begin{figure*}
  \centering
  \hfill%
  \begin{subfigure}[b]{0.305\textwidth}
    \centering\includegraphics[width=0.9\textwidth]{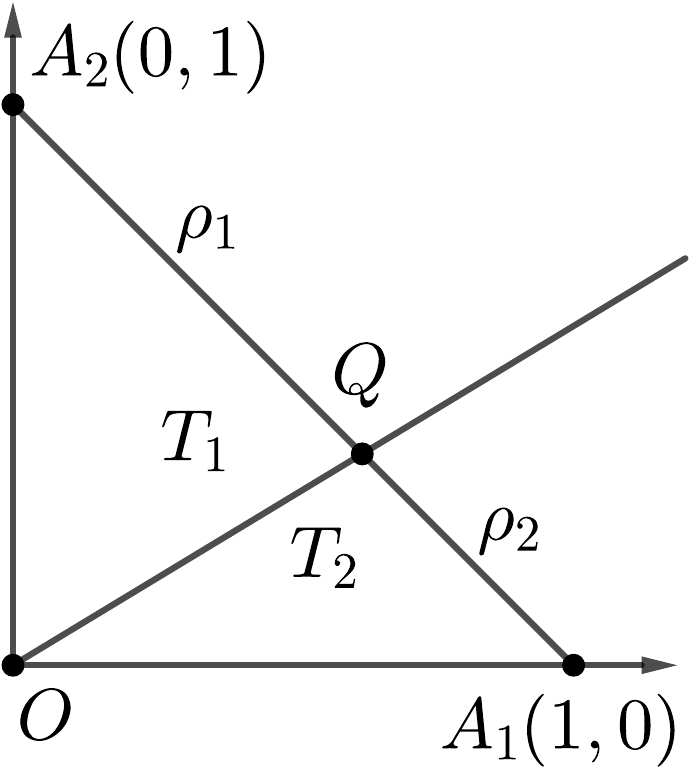}
    \caption{The two-player case.}\label{sfig:2simplex}
  \end{subfigure}\hfill%
  \begin{subfigure}[b]{0.400\textwidth}
    \centering\includegraphics[width=0.9\textwidth]{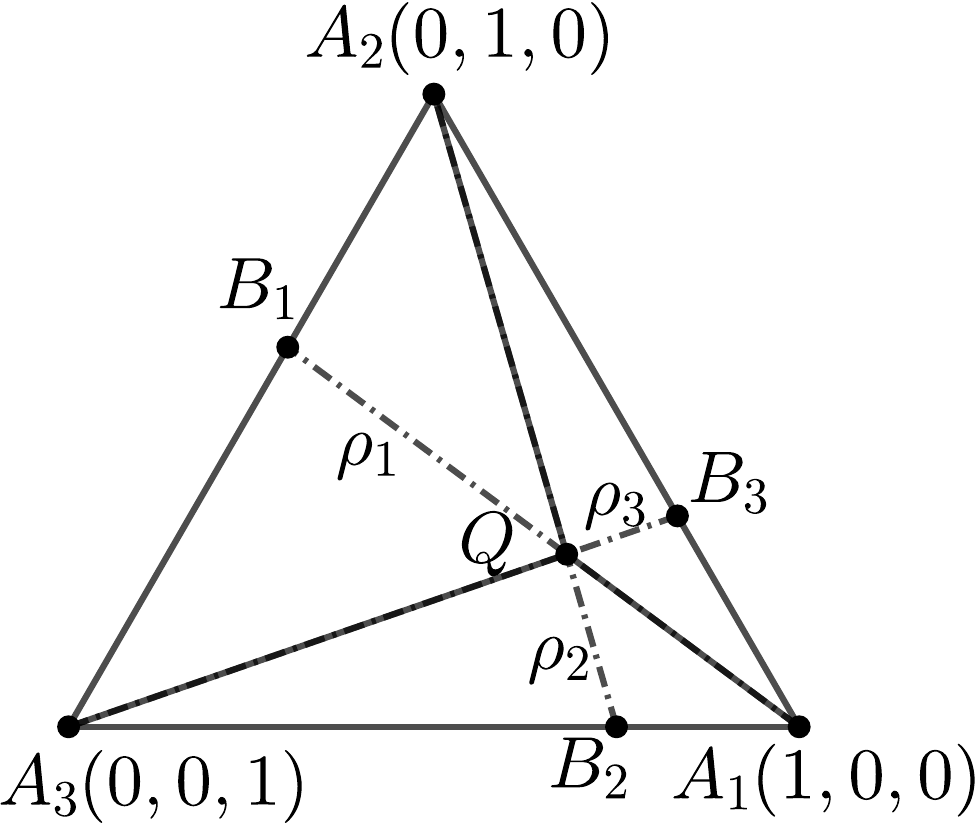}
    \caption{The three-player case.}\label{sfig:3simplex}
  \end{subfigure}\hfill.
  \caption{Simplexes for two-player and three-player cases.}
  \label{fig:simplex}
\end{figure*}

\bigskip
For the two-player case (see \autoref{sfig:2simplex}), by definition,
\begin{gather*}
  T_1 = \{\t \in [0, 1]^2 : t_1 / x_1 \leq t_2 / x_2\},  \\
  T_2 = \{\t \in [0, 1]^2 : t_2 / x_2 \leq t_1 / x_1\}.
\end{gather*}
Therefore, the formula of the hyperplane (in fact a line in the $2$-dimensional
case) separating $T_1$ and $T_2$ is $t_1 / x_1 = t_2 / x_2$ and $Q$ is the
intersection of this line and the simplex $A_1A_2$. In particular, the
$1$-dimensional volume of the set $\{\t \in T_1 : \bar t = 1\}$ is exactly the
length of the line segment $QA_2$, which equals to $\rho_1(1) = \sqrt2 x_1 /
(x_1 + x_2)$. Similarly, we get $\rho_2(1) = \sqrt2 x_2 / (x_1 + x_2)$.

\bigskip
For the three-player case (see \autoref{sfig:3simplex}), by definition,
\begin{gather*}
  T_1 = \{\t \in [0, 1]^3 : t_1 / x_1 \leq t_2 / x_2, t_1 / x_1 \leq t_3 / x_3\},  \\
  T_2 = \{\t \in [0, 1]^3 : t_2 / x_2 \leq t_1 / x_1, t_2 / x_2 \leq t_3 / x_3\},  \\
  T_3 = \{\t \in [0, 1]^3 : t_3 / x_3 \leq t_1 / x_1, t_3 / x_3 \leq t_2 / x_2\}.
\end{gather*}
In this case, there are three hyperplanes separating $T_1$, $T_2$, and $T_3$:
$t_1 / x_1 = t_2 / x_2$, $t_2 / x_2 = t_3 / x_3$, and $t_3 / x_3 = t_1 / x_1$.
Note that the intersection of the three hyperplanes is a line, $t_1 / x_1 = t_2
/ x_2 = t_3 / x_3$, which then intersects with the simplex $A_1A_2A_3 = \{\t:
t_1 + t_2 + t_3 = 1\}$ at a point $Q$.

Therefore, the $2$-dimensional volumes (or areas) $\rho_1(1)$, $\rho_2(1)$, and
$\rho_3(1)$ are exactly the areas of triangles $\triangle QA_2A_3$, $\triangle
A_1QA_3$, and $\triangle A_1A_2Q$, respectively. Note that $A_1A_2A_3$ is a
regular $2$-simplex, hence $A_2A_3 = A_3A_1 = A_1A_2$. Thus the ratios between
the areas of the triangles equal to the ratios between the heights of the
triangles, because their bases ($A_2A_3$, $A_3A_1$, $A_1A_2$) are of the same
length.

We conclude that
\begin{gather*}
  \frac{\rho_1(1)}{\rho_2(1)} = \frac{d(Q, A_2A_3)}{d(Q, A_3A_1)},
  \frac{\rho_1(1)}{\rho_3(1)} = \frac{d(Q, A_2A_3)}{d(Q, A_1A_2)},
\end{gather*}
where $d$ denotes the distance between a point and a line.

In the meanwhile, recall that the coordinate of point $Q$ is:
\[
  \left(\frac{x_1}{x_1 + x_2 + x_3}, \frac{x_2}{x_1 + x_2 + x_3},
  \frac{x_3}{x_1 + x_2 + x_3}\right),
\]
and with some primary calculation, we know that
\[
  \frac{d(Q, A_2A_3)}{d(Q, A_3A_1)} = \frac{x_1}{x_2} \quad \text{and} \quad
  \frac{d(Q, A_2A_3)}{d(Q, A_1A_2)} = \frac{x_1}{x_3}.
\]
In other words, $\rho_i(1) / \rho(1) = x_i / (x_1 + x_2 + x_3)$.

\bigskip
Finally, we generalize the observations above to get the proof.
\begin{proof}[Proof of \autoref{lem:splxv}]
  By the definition, $T_1, \ldots, T_m$ are separated by ${m \choose 2}$
  hyperplanes:
  \begin{align*}
    \forall i\neq j, t_i / x_i = t_j / x_j.
  \end{align*}
  In particular, the intersection of them is a line and the line intersects the
  $(m-1)$-simplex $A_1\cdots A_m$ at a unique point $Q$, whose coordinate is:
  \begin{align*}
    Q: \left(\frac{x_1}{x_1 + \cdots + x_m}, \ldots,
    \frac{x_m}{x_1 + \cdots + x_m}\right).
  \end{align*}

  In the meanwhile, note that for each $i \in [m]$, point $A_j \in T_i$ for all
  $j \neq i$. Therefore, the set $\{\t \in T_i : \bar t = 1\}$ is a
  $(m-1)$-simplex, $A_1\cdots A_{i-1}QA_{i+1}\cdots A_m$ and its volume is
  proportional to the $(m-2)$-dimensional volume of its base
  $A_1\cdots A_{i-1}A_{i+1}\cdots A_m$ times its height (the distance between
  point $Q$ and the base).

  Because $A_1\cdots A_m$ is a regular $(m-1)$-simplex, the volumes of the bases
  for each $i$ are the same. Therefore, $\rho_i(1) / \rho_j(1)$ equals to the
  ratios between the distances from $Q$ to the corresponding bases. Based on the
  coordinate of $Q$, we conclude that the ratio equals to $x_i / x_j$. In other
  words,
  \begin{align*}
    \frac{\rho_i(1)}{\rho(1)}
      = \frac{\rho_i(1)}{\rho_1(1) + \cdots + \rho_m(1)}
      = \frac{x_i}{x_1 + \cdots + x_m}.
  \end{align*}
\end{proof}

\section{Missing Proofs}\label{sec:proof}

The following technical lemma is repeatedly used in the proofs.
\begin{lemma}\label{lem:tech}
  For all $x \in \R$, $e^x \geq 1 + x$.
\end{lemma}

\begin{proof}[Proof of \autoref{lem:core}]
  For the $\core$ part, because $\bar t < \alpha < n^{-2/3}$,
  \begin{gather*}
    \prod_{i=1}^m (1-t_i)
      = 1 - \sum_i t_i + \sum_{i \neq j} t_i t_j + \cdots
      = 1 - \bar t + O(\bar t^2),  \\
    \prod_{i=1}^m (1-t_i)^{n-1}
      = (1 - \bar t)^{n-1}\left(1 + \frac{O(\bar t^2)}{1 - \bar t}\right)^{n-1}.
  \end{gather*}
  Since $n \bar t^2 \leq n \alpha^2 < n^{-1/3}$, then by \autoref{lem:tech} and
  Taylor's expansion of $e^x$ at $0$,
  \begin{align*}
    \prod_{i=1}^m (1-t_i)^{n-1}
      \leq (1 - \bar t)^{n-1} e^{\frac{O(\bar t^2)}{1 - \bar t} \cdot (n - 1)}
    ={}(1 - \bar t)^{n-1}
      \cdot \left(1 + \frac{O(\bar t^2)(n - 1)}{1 - \bar t} + O(n^2\bar t^4)\right)
    ={} (1 - \bar t)^{n-1} + O(n\bar t^2).
  \end{align*}
  Thus
  \begin{align*}
    \int_{\core(T_i)} f(\t) \ud \t
      \leq \int_{\core(T_i)} n^m\left((1 - \bar t)^{n-1} + O(n\bar t^2)\right) \ud \t,
  \end{align*}
  where
  \begin{align*}
    \int_{\core(T_i)} n^m \cdot O(n\bar t^2) \ud \t
      \leq n^{m+1} \int_\core O(\bar t^2) \ud t.
  \end{align*}
  Consider rewrite the right-hand-side as a double integration over $\core_{z} =
  \{\t \in \core : \bar t = z\}$ and then for $z$ from $0$ to $\alpha$:
  \begin{align*}
    \int_\core O(\bar t^2) \ud t
      = \int_0^\alpha O(z^2) \ud z \int_L \ud {\core_{z}}
  \end{align*}

  Note that $\int_L \ud {\core_{z}}$ is the Lebesgue integration over the set
  $\core_{z}$ and hence equals to its $(m-1)$-dimensional volume, which is
  $\rho(z) = \frac{\sqrt{m}}{(m-1)!} \cdot z^{m-1}$ \citep{stein1966note}.
  Therefore,
  \begin{align*}
    \int_\core O(\bar t^2) \ud t
      = \int_0^\alpha O(z^2) \cdot z^{m-1} \ud z = O(\alpha^{m+2}).
  \end{align*}
  That is
  \begin{align*}
    \int_{\core(T_i)} n^m \cdot O(n\bar t^2) \ud \t \leq O(n^{m+1}\alpha^{m+2}).
  \end{align*}
  In summary,
  \begin{align*}
    \int_{\core(T_i)} f(\t) \ud \t
      = \int_{\core(T_i)} g(\t) \ud \t + O(n^{m+1}\alpha^{m+2}).
  \end{align*}
\end{proof}

\begin{proof}[Proof of \autoref{lem:tail}]
  For the $\tail$ part, according to the inequality of arithmetic and geometric
  means,
  \begin{align*}
    \prod_{i=1}^m (1-t_i) \leq (1 - \bar t / m)^m,
  \end{align*}
  then by \autoref{lem:tech} and $\bar t \geq \alpha$,
  \begin{align*}
    \prod_{i=1}^m (1-t_i)^{n-1} \leq (1 - \bar t / m)^{m(n-1)}
      \leq e^{-\bar t(n - 1)} \leq e^{-\alpha(n - 1)}.
  \end{align*}
  Therefore, $f(\t) \leq n^m e^{-\alpha(n - 1)} = e^{-\alpha(n - 1) + m \ln n}$.
  Also note that $\int_{\tail} \ud \t$ is the volume of $\tail$, which is less
  than $1$, hence
  \begin{align*}
    \int_{\tail(T_i)} f(\t) \ud \t
      \leq e^{-\alpha(n - 1) + m \ln n} \int_{\tail} \ud \t
      \leq e^{-\alpha(n - 1) + m \ln n}.
  \end{align*}
\end{proof}

\bibliographystyle{apalike}
\balance
\bibliography{mining}

\end{document}